\documentclass[journal,onecolumn]{IEEEtran}
\usepackage{booktabs}
\usepackage{amsmath}
\usepackage{amsmath}                
\usepackage{mathtools} 
\usepackage{amsthm}                 
\usepackage{amssymb}
\usepackage{thmtools}
\usepackage{kpfonts}
\usepackage[geometry]{ifsym}
\usepackage{dsfont}
\usepackage{multirow}
\usepackage[mathscr]{euscript}
\usepackage{graphicx}
\usepackage{color}
\usepackage[inline]{enumitem}
\usepackage{framed}
\usepackage{float}
\usepackage{longtable}
\usepackage{cite}
\usepackage{enumitem}
\usepackage{caption} 
\usepackage{subcaption} 
\usepackage{xcolor}

\usepackage[colorlinks=true, citecolor=blue, linkcolor=blue]{hyperref}       
\usepackage[font=itshape]{quoting}
\usepackage{lscape}
\usepackage{makecell}
\usepackage[section]{placeins}

\usepackage{listings}
\usepackage{comment}
\usepackage{framed} 
\usepackage{nomencl} 
\makenomenclature

\declaretheoremstyle[%
  headfont=\bfseries,%
  headpunct={:},%
  notefont=\normalfont\bfseries,%
  notebraces={--~}{},
    qed=$\blacksquare$,
]{definitionstyle}
\theoremstyle{definition}
\declaretheorem[style=definitionstyle,name=Definition]{defn}
\declaretheorem[style=definitionstyle,name=Theorem]{thm}

\theoremstyle{definition}

\theoremstyle{plain}
\theoremstyle{remark}

\begin{document}
%
\title{Enhancing Spatio-Temporal Resolution of Process-Based Life Cycle Analysis with Model-Based Systems Engineering \&
Hetero-functional Graph Theory}

\author{Niraj Gohil, Nawshad Haque, Amgad Elgowainy, Amro M. Farid}


\date{March 28 2024}
\maketitle

\begin{abstract}
Life cycle analysis (LCA) has emerged as a vital tool for assessing the environmental impacts of products, processes, and systems throughout their entire lifecycle. It provides a systematic approach to quantifying resource consumption, emissions, and waste, enabling industries, researchers, and policymakers to identify hotspots for sustainability improvements. By providing a comprehensive assessment of systems, from raw material extraction to end-of-life disposal, LCA facilitates the development of environmentally sound strategies, thereby contributing significantly to sustainable engineering and informed decision-making. Despite its strengths and ubiquitous use, life cycle analysis has not been reconciled with the broader literature in model-based systems engineering and analysis, thus hindering its integration into the design of complex systems more generally. This lack of reconciliation poses a significant problem, as it hinders the seamless integration of environmental sustainability into the design and optimization of complex systems. Without alignment between life cycle analysis (LCA) and model-based systems engineering (MBSE), sustainability remains an isolated consideration rather than an inherent part of the system’s architecture and design. The original contribution of this paper is twofold. First, the paper reconciles process-based life cycle analysis with the broader literature and vocabulary of model-based systems engineering and hetero-functional graph theory. It ultimately proves that model-based systems engineering and hetero-functional graph theory are a formal generalization of process-based life cycle analysis. Secondly, the paper demonstrates how model-based systems engineering and hetero-functional graph theory may be used to enhance the spatio-temporal resolution of process-based life cycle analysis in a manner that aligns with system design objectives.  
\end{abstract}

\section{Introduction}
Life cycle analysis (LCA) has emerged as an indispensable tool for evaluating the environmental impacts of products, processes, and systems throughout their entire life cycle\cite{Curran1994}. It provides a systematic approach to quantifying resource consumption, emissions, and waste, enabling industries, researchers, and policymakers to identify hotspots for sustainability improvements. By offering a comprehensive assessment of systems from raw material extraction to end-of-life disposal, LCA supports the development of environmentally sound strategies, contributing significantly to sustainable engineering and decision-making\cite{Rebitzer2004}.  Over time, LCA has evolved to address challenges such as resource scarcity and environmental degradation.  It has also supported boundary identification, allocation methods, and impact assessment techniques \cite {Shi2015}.  Emerging innovations, including hybrid models and real-time data integration, further enhance LCA's robustness, enabling stakeholders to navigate the complexities of sustainable engineering while fostering innovation and aligning with global sustainability goals. This holistic perspective empowers stakeholders to identify processes that consume resources and produce emissions, and then implement targeted interventions that drive sustainability forward \cite{daCosta2024}.

Despite its strengths and widespread use, life cycle analysis has not been reconciled with the broader literature in model-based systems engineering and analysis \cite{Madni2018ModelbasedSE}, thereby hindering its integration into the design of complex systems more generally. This lack of reconciliation poses a significant problem, as it hinders the seamless integration of environmental sustainability into the design and optimization of complex systems. Without alignment between life cycle analysis (LCA) and model-based systems engineering (MBSE), sustainability remains an isolated consideration rather than an inherent part of the system's architecture and design\cite{
Mooij2012}.  This isolation leads to missed opportunities to identify and address environmental trade-offs during the critical early stages of design, where decisions have the greatest impact on a system's life cycle. Additionally, the inability to evaluate environmental performance alongside technical, operational, and economic objectives results in fragmented decision-making, where environmental considerations may conflict with other system goals. Ultimately, this disconnect limits the development of truly holistic solutions that balance sustainability with the multifaceted demands of modern complex systems, hindering progress toward resilient and environmentally sound innovations.

\vspace{-0.1in}
\subsection{Original Contribution}
The original contribution of this paper is twofold.  First, the paper seeks to reconcile process-based life cycle analysis with the broader literature and vocabulary of model-based systems engineering\cite{Delligatti2013SysMLDA} and hetero-functional graph theory\cite{Schoonenberg:2019:ISC-BK04, Farid:2022:ISC-J51}. Doing so reveals the specific limiting conditions upon which model-based systems engineering (MBSE) and hetero-functional graph theory (HFGT) collapse to a life cycle analysis.  Therefore, this paper proves that model-based systems engineering and hetero-functional graph theory are a formal generalization of process-based life cycle analysis.  Consequently, this paper discusses how a life cycle analysis methodology can be extended when these specific limiting conditions are relaxed.  In particular, it shows how MBSE and HFGT address systems that a.) contain processes with a diversity of processing times, b.) exhibit dynamic behavior over a long simulation horizon, c.) store physical quantities, and d.) exhibit an explicit spatial distribution.  Secondly, the paper demonstrates how model-based systems engineering and hetero-functional graph theory may be used to enhance the spatio-temporal resolution of process-based life cycle analysis in a manner that aligns with system design objectives.  

\vspace{-0.1in}
\subsection{Paper Outline}
The remainder of the paper proceeds as follows.  Section \ref{Sec:Background} provides preliminary background on process-based life cycle analysis, model-based systems engineering, and hetero-functional graph theory.  Section \ref{Sec:III AnalysisExample} then reconciles these three modeling and analysis techniques.  Ultimately, it shows that MBSE and HFGT are a formal generalization of process-based life cycle analysis.  Sec. \ref{Sec:Discussion} then discusses how this MBSE and HFGT can be used to extend the capabilities of process-based life cycle analysis.  Finally, Sec. \ref{Sec:Conclusion} brings the work to a conclusion.  

\section{Background}\label{Sec:Background}
In order to support the analytical discussion in the following sections, this section provides preliminary background on process-based life cycle analysis in Sec. \ref{Sec:LCAIntro}, on model-based systems engineering in Sec. \ref{Sec:MBSEIntro}, and on hetero-functional graph theory in Sec. \ref{Sec:HFGTIntro}.  

\vspace{-0.1in}
\subsection{Process Based Life Cycle Analysis}\label{Sec:LCAIntro}
Life cycle analysis has a rich and well-regarded literature that spans decades of research, evaluating the environmental impacts of products, processes, and systems across their entire lifecycle \cite{klopffer2014}. It has been widely used in sectors such as manufacturing, energy, and transportation, where it helps design more sustainable systems and informs regulatory compliance, policy development, and corporate sustainability strategies \cite{Jaeger2018LCAIS}.  Among its various methodologies, process-based Life Cycle Analysis (LCA) stands out for its extensible and systematic approach. It breaks down the lifecycle of a product or system into a number of process stages such as raw material extraction, manufacturing, use, and disposal \cite{tillman_methodology}. The representation of systems in terms of their constituent processes allows for granular analysis of resource consumption, air and water emissions, and waste generation\cite{villeneuve2007waste}.  The quantification of these values serves to identify where to make the most important sustainability improvements, making it an indispensable tool for industries, researchers, and policymakers  \cite {toniolo2021lca}.

Process-based life cycle analysis presumes a process flow diagram that includes processes and their associated inputs and outputs \cite{hendrickson_economic_io_lca}. An example process flow diagram is shown in Fig. \ref{fig: Process Flow Diagram} (and is elaborated further in Sec. \ref{Sec:LCAExample}).  These processes are connected to each other with directed arrows in parallel and serial arrangements to describe the behavior of the system or product as a whole. Importantly, each arrow is labeled with a numerical weight that defines the amount of output produced relative to the amount of input consumed.  Each process represents a distinct stage in the lifecycle, such as material extraction, manufacturing, use, or disposal. Meanwhile, the connecting arrows indicate the transfer of resources or emissions between these process stages. The numerical weights on the arrows often represent stoichiometric ratios of various chemical/conversion processes. By clearly defining the interactions between processes, inputs, and outputs, the process flow diagram facilitates the identification of key leverage points for sustainability improvements within the system\cite{mathews_lca_decisions}.

\begin{figure}[H]
\centering
\includegraphics[width=1\linewidth]{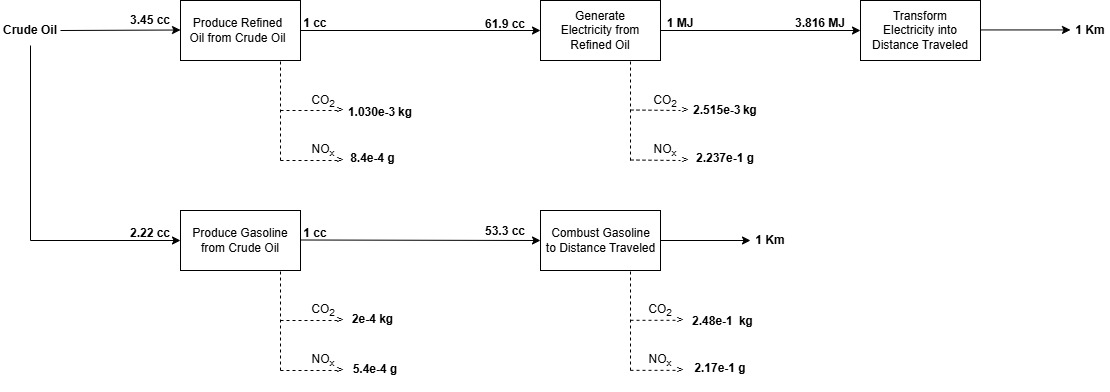}
\caption{A process flow diagram illustrating the conversion of crude oil into vehicle motion for electric vehicles (EVs) and internal combustion vehicles (ICVs). It outlines the oil refining, electricity generation, gasoline production, electro-mechanical conversion, and combustion processes required for vehicle motion. CO$_2$ and NO$_X$ emissions have been identified at each stage as relevant environmental aspects.}
\label{fig: Process Flow Diagram}
\end{figure}

The visual representation of a process flow diagram provides the basis for quantitative analysis of the flows of matter, energy, and emissions using linear algebra. Given a set of processes P in the process flow diagram, a vector X describes how much these processes are executed. In the case of discrete processes, X corresponds to the number of times the process is executed. In the case of continuous processes, X corresponds to how many time steps the process is executed. Next, each process is assigned a single primary product. The associated vector Y measures the amount of these products produced when the processes are executed once (or for one time step). Next, a square technology matrix A is defined. Its elements $a_{ij}$ are coefficients that detail the amount of product $y_i$ produced when process $p_j$ is executed once (or for one time step). Note that when a process $p_j$ consumes a certain product $y_i$, the associated element $a_{ij}$ takes on a negative value by convention. Also note that the number of processes and products is assumed to be the same. Therefore, the technology matrix A is always square. Consequently, 
\begin{align}
Y=AX
\label{Equation 1}
\end{align}
Additionally, a vector of environmental aspects (i.e., environmental emissions or resource impacts) $E$ is defined.  Then, an environmental matrix $B$ is defined.\cite{mathews_lca_decisions} Its elements $b_{kj}$ are coefficients that detail the amount of environmental aspect produced $e_k$ as a result of executing process $p_j$ once. Again, when a process $p_j$ \emph{consumes} a certain environmental aspect $e_i$, the associated element $b_{ij}$ takes on a negative value by convention.\cite{mathews_lca_decisions}.  Therefore,
\begin{align}
E=BX
\label{Equation 2}
\end{align}
Finally, the process-based life cycle analysis is executed by solving for $E$ in terms of $Y$ \cite{mathews_lca_decisions}.
\begin{align}
E = BA^{-1}Y
\label{
Equation 3}
\end{align}

\subsection{Model-Based Systems Engineering}\label{Sec:MBSEIntro}
In the meantime, Model-Based Systems Engineering (MBSE) has emerged as a transformative approach for graphically modeling large, complex systems throughout their lifecycle\cite{Bajaj2011SLIMCM}. It uses the Systems Modeling Language (SysML) to describe a system's requirements, structure, behavior, and their interdependencies\cite{Kaslow2016CubeSatMB}. MBSE has become increasingly important due to its ability to enhance traceability, improve stakeholder communication, and provide real-time analysis capabilities.  These characteristics have made it particularly valuable in managing complexity in fields such as aerospace, automotive, defense, telecommunication, and energy \cite {Madni2018ModelbasedSE}. MBSE has been widely used to design complex systems, optimize operations, and ensure compliance with regulatory requirements by integrating the applicable required metrics into system design, operation, and maintenance \cite{Dean2012ModelBasedSE}.

While a complete introduction to Model-Based Systems Engineering (MBSE) and SysML is beyond the scope of this paper, the essential elements of two key SysML diagrams are introduced to understand how to represent a system’s architecture\cite{Crawley2015SystemAS}. The first diagram, called the Block Definition Diagram (BDD), is used to model a system’s form or its component parts.  An example BDD is shown in Fig.\ref{fig: BDD to Oil Vehicle Motion}
 (and is elaborated further in Sec.\ref{Sec:III AnalysisExample}). In the context of this paper, a BDD contains (at a minimum): 

\begin{itemize}
\item \textbf{Blocks:} that represent system elements, components, or subsystems.
\item \textbf{Attributes:} that represent characteristics or properties of each block.
\item \textbf{Operations:} that represent functions, activities, or processes that a block can perform.  Each combination of an operation in a block describes a capability.  
\end{itemize}

\begin{figure}[H]
\centering
\includegraphics[
width=1\linewidth]{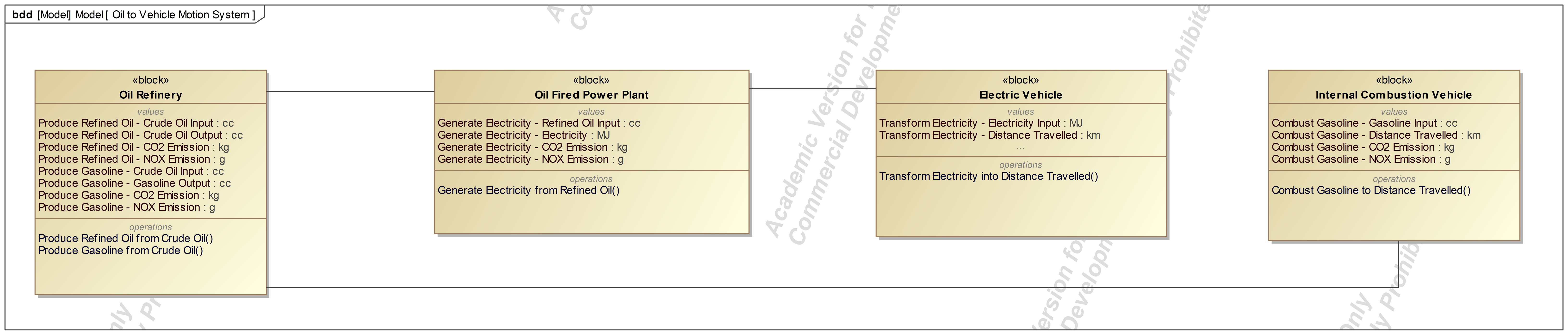}
\caption{A block definition diagram illustrating the Oil to Vehicle Motion System model, breaking down the key components and flow properties. The system encompasses the oil refinery, oil-fired power plant, electric vehicle (EV), and internal combustion engine (ICV). The diagram outlines each block's input and output parameters, including crude oil, refined oil, electricity, gasoline, emissions (CO$_2$ and NO$_x$), and vehicle range in kilometers, showcasing the interconnected pathways for energy conversion and emissions output.}
\label{fig: BDD to Oil Vehicle Motion}
\end{figure}

The second diagram, called the Activity Diagram (ACT), is used to model a system's function in terms of its constituent activities. An example ACT is shown in Fig.~\ref{fig: Activity Diagram Propel Vehicle from Crude Oil} (and is elaborated further in Sec.~\ref {Sec:III AnalysisExample}). In the context of this paper, an ACT contains (at a minimum):

\begin{figure}[H]
\centering
\includegraphics[width=1\linewidth]{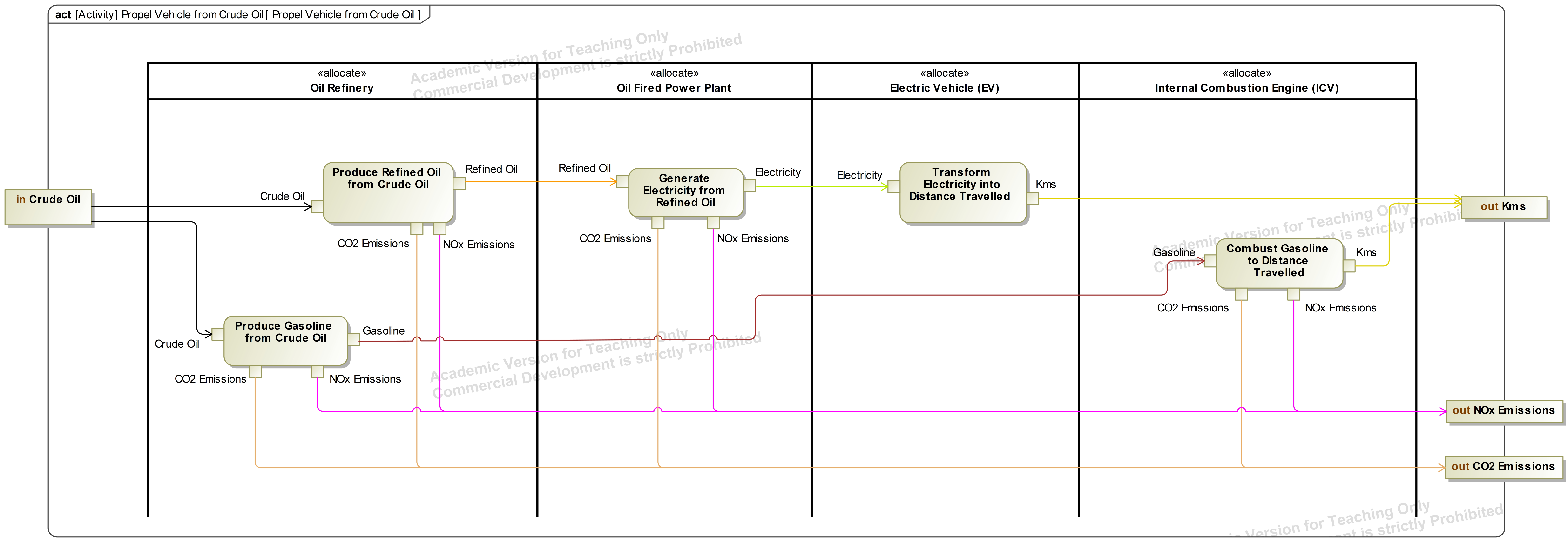}
\caption{An activity diagram with swim lanes illustrating the refining of oil, the production of gasoline, the generation of electricity, the transformation of electricity into motion, and the combustion of gasoline for motion.  Each swim lane represents a system component responsible for carrying out each action.  The activity diagram also shows crude oil as the input, and distance traveled,  CO$_2$ emissions, and NO$_X$ emissions as the three outputs.}
\label{fig: Activity Diagram Propel Vehicle from Crude Oil}
\end{figure}

\begin{itemize}
\item \textbf{Actions (or Activities)}: Represent specific tasks, operations, or processes within the system.  Note that the operations found in a BDD are often represented as actions in an ACT.  Similarly, there is no distinction between a process in a process flow diagram and an action in an ACT.  
\item \textbf{Transitions}: Arrows indicate the flow from one action to another, showing the sequential or parallel execution of processes.
\item \textbf{Swim Lanes}: Horizontal or vertical partitions that separate the activities performed by different system components or roles, clarifying responsibilities.  The presence of an action within a swim lane also describes a capability.  
\end{itemize}
Together, the activity diagram illustrates how the system's functions are executed to reveal its behavior.  

\subsection{Hetero-functional Graph Theory}\label{Sec:HFGTIntro}
While MBSE, and more specifically SysML, can graphically model large, complex systems, it does not have in-built functionality for conducting quantitative analysis.   Fortunately, hetero-functional graph theory (HFGT) provides an analytical means of translating graphical SysML models into mathematical models.  As shown in Fig. \ref{Fig:LFESMetaArchitecture}, this translation requires the HFGT meta-architecture stated in SysML.  
\begin{figure}[htbp]
\centering
\includegraphics[width=\textwidth]{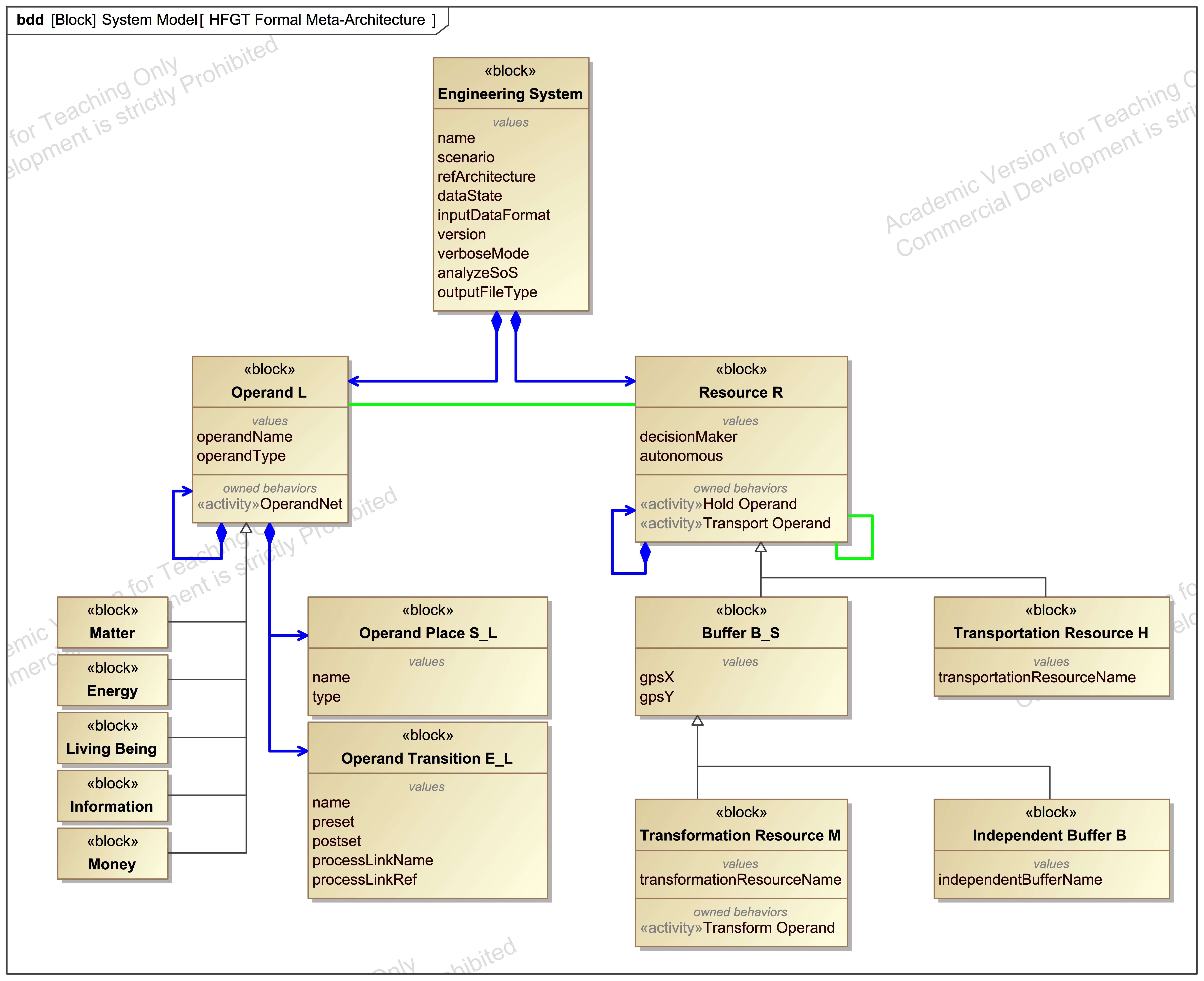}
\caption{A SysML Block Definition Diagram of the System Form of the Engineering System Meta-Architecture\cite{Schoonenberg:2019:ISC-BK04}.}
\label{Fig:LFESMetaArchitecture}
\end{figure}

The HFGT meta-architecture introduces several meta-elements whose definitions are formally introduced here:  
\begin{defn}[System Operand \cite{SE-Handbook-Working-Group:2015:00}]\label{Defn:D1 System Operand}
An asset or object $l_i \in L$ that is operated on or consumed during the execution of a process.
\end{defn}
\begin{defn}[System Process\cite{Hoyle:1998:00,SE-Handbook-Working-Group:2015:00}]\label{Defn:D2 System Process}
An activity $p \in P$ that transforms or transports a predefined set of input operands into a predefined set of outputs. 
\end{defn}
\begin{defn}[System Resource \cite{SE-Handbook-Working-Group:2015:00}]\label{Defn:D3 System Resource}
An asset or object $r_v \in R$ that facilitates the execution of a process.  
\end{defn}
\noindent Importantly, these meta-elements are organized around the universal structure of human language.  Namely, system resources $R$ serve as subjects, system processes $P$ serve as predicates, and operands $L$ serve as objects within the predicates.  

The system resources $R=M \cup B \cup H$ are classified into transformation resources $M$, independent buffers $B$, and transportation resources $H$.  Additionally, the set of ``buffers" $ B_S=M\cup B$ is introduced to support the discussion.  
\begin{defn}[Buffer\cite{Schoonenberg:2019:ISC-BK04,Farid:2022:ISC-J51}]\label{Defn:D4 Buffer}
A resource $r_v \in R$ is a buffer $b_s \in B_S$ if it is capable of storing or transforming one or more operands at a unique location in space.  
\end{defn}
Equally important, the system processes $P = P_\mu \cup P_{\bar{\eta}}
$ are classified into transformation processes $P_\mu$ and refined transportation processes $P_\eta
$.  The latter arises from the simultaneous execution of one transportation process and one holding process.  Finally, hetero-functional graph theory emphasizes that resources are capable of one or more system processes to produce a set of ``capabilities"\cite{Schoonenberg:2019:ISC-BK04}.
\begin{defn}[Capability\cite{Schoonenberg:2019:ISC-BK04,Farid:2022:ISC-J51,Farid:2016:ISC-BC06}]\label{Defn:D5 Capability}
An action $e_{wv} \in {\cal E}_S$ (in the SysML sense) defined by a system process $p_w \in P$ being executed by a resource $r_v \in R$.  It constitutes a subject + verb + operand sentence of the form: ``Resource $r_v$ does process $p_w
$".  
\end{defn}
\noindent The highly generic and abstract nature of these definitions has allowed HFGT to be applied to numerous application domains, including electric power, potable water, wastewater, natural gas, oil, coal, multi-modal transportation, mass-customized production, and personalized healthcare delivery systems.  For a more in-depth description of HFGT, readers are directed to past works\cite{Schoonenberg:2019:ISC-BK04,Farid:2022:ISC-J51,Farid:2016:ISC-BC06}.

Returning to Fig. \ref{Fig:LFESMetaArchitecture}, the engineering system meta-architecture stated in SysML must be instantiated and ultimately transformed into the associated Petri net model. To that end, the positive and negative hetero-functional incidence tensors (HFIT) are introduced to describe the flow of operands through buffers and capabilities.  
\begin{defn}[The Negative 3$^{rd}$ Order Hetero-functional Incidence Tensor (HFIT) $\widetilde{\cal M}_\rho^-$\cite{Farid:2022:ISC-J51}]\label{Defn:D6 HFIT -ve}
The negative hetero-functional incidence tensor $\widetilde{\cal M_\rho}^- \in \{0,1\}^{|L|\times |B_S| \times |{\cal E}_S|}$  is a third-order tensor whose element $\widetilde{\cal M}_\rho^{-}(i,y,\psi)=1$ when the system capability ${\epsilon}_\psi \in {\cal E}_S$ pulls operand $l_i \in L$ from buffer $b_{s_y} \in B_S$.
\end{defn} 
\begin{defn}[The Positive  3$^{rd}$ Order Hetero-functional Incidence Tensor (HFIT)$\widetilde{\cal M}_\rho^+$\cite{Farid:2022:ISC-J51}]\label{Defn:D7 HFIT +ve}
 The positive hetero-functional incidence tensor $\widetilde{\cal M}_\rho^+ \in \{0,1\}^{|L|\times |B_S| \times |{\cal E}_S|}$  is a third-order tensor whose element $\widetilde{\cal M}_\rho^{+}(i,y,\psi)=1$ when the system capability ${\epsilon}_\psi \in {\cal E}_S$ injects operand $l_i \in L$ into buffer $b_{s_y} \in B_S$.
\end{defn}
\noindent These incidence tensors are straightforwardly ``matricized" to form the 2$^{nd}$ Order Hetero-functional Incidence Matrix $M = M^+ - M^-$ with dimensions $|L||B_S|\times |{\cal E}|$. Consequently, the supply, demand, transportation, storage, transformation, assembly, and disassembly of multiple operands in distinct locations over time can be described by an Engineering System Net and its associated State Transition Function\cite{Schoonenberg:2022:ISC-J50}.
\begin{defn}[Engineering System Net\cite{Schoonenberg:2022:ISC-J50}]\label{Defn:8 ESN}
An elementary Petri net ${\cal N} = \{S, {\cal E}_S, \textbf{M}, W, Q\}$, where
\begin{itemize}
\item $S$ is the set of places with size: $|L||B_S|$,
\item ${\cal E}_S$ is the set of transitions with size: $|{\cal E}|$,
\item $\textbf{M}$ is the set of arcs, with the associated incidence matrices: $M = M^+ - M^-$,
\item $W$ is the set of weights on the arcs, as captured in the incidence matrices,
\item $Q=[Q_B; Q_E]$ is the marking vector for both the set of places and the set of transitions. 
\end{itemize}
\end{defn}
\begin{defn}[Engineering System Net State Transition Function\cite{Schoonenberg:2022:ISC-J50}]\label{Defn:9 ESN-STF}
The  state transition function of the engineering system net $\Phi()$ is:
\begin{equation}\label{CH6:eq:PhiCPN}
Q[k+1]=\Phi(Q[k],U^-[k], U^+[k]) \quad \forall k \in \{1, \dots, K\}
\end{equation}
where $k$ is the discrete time index, $K$ is the simulation horizon, $Q=[Q_{B}; Q_{\cal E}]$, $Q_B$ has size $|L||B_S| \times 1$, $Q_{\cal E}$ has size $|{\cal E}_S|\times 1$, the input firing vector $U^-[k]$ has size $|{\cal E}_S|\times 1$, and the output firing vector $U^+[k]$ has size $|{\cal E}_S|\times 1$.  
\begin{align}\label{Eq:ESNSTF1}
Q_{B}[k+1]&=Q_{B}[k]+{M}^+U^+[k]\Delta T-{M}^-U^-[k]\Delta T \\ \label{Eq:ESNSTF2}
Q_{\cal E}[k+1]&=Q_{\cal E}[k]-U^+[k]\Delta T +U^-[k]\Delta T
\end{align}
where $\Delta T$ is the duration of the simulation time step.  
\end{defn}
When the capabilities are assumed to occur instantaneously, $U^+[k]=U^-[k] \forall k$.  Then, Eq. \ref{CH6:eq:PhiCPN} collapses to triviality and Eq. \ref{Eq:ESNSTF2} becomes
\begin{align}\label{Eq:SimpleSTF}
Q_{B}[k+1]&=Q_{B}[k]+{M}U[k]
\Delta T 
\end{align}
Ultimately, the reconciliation of process-based life cycle analysis, model-based systems engineering, and hetero-functional graph theory relies on comparing the mathematics of process-based life cycle analysis with the engineering system net state transition function found in Defn. \ref{Defn:9 ESN-STF}. 

\section{Reconciliation of Process-based Life Cycle Analysis, Model-Based Systems Engineering, \& Hetero-functional Graph Theory}\label{Sec:III AnalysisExample}

While it is clear that process-based life cycle analysis, model-based systems engineering, and hetero-functional graph theory can all model a diversity of complex systems,  it is also clear that they use significantly different terminology that makes it difficult to relate them to each other conceptually. To demonstrate these relationships concretely, the illustrative example first introduced in Figures \ref{fig: Process Flow Diagram}, \ref{fig: BDD to Oil Vehicle Motion} and \ref{fig: Activity Diagram Propel Vehicle from Crude Oil} is discussed from the perspective of process-based life cycle analysis in Sec. \ref{Sec:LCAIntro} and then next from the perspective of hetero-functional graph theory in Sec. \ref{Sec:HFGTIntro}.  

\subsection{Process-Based Life Cycle Analysis}\label{Sec:LCAExample}
To demonstrate process-based life cycle analysis concretely, the process flow diagram shown in Fig. \ref{fig: Process Flow Diagram} is taken as an illustrative example.  It shows that the motion of an internal combustion vehicle (ICV) is produced from the combustion of gasoline, which in turn is produced from crude oil.   Along the way, both processes produce carbon dioxide ($CO_2$) and nitrogen oxides ($NO_X$) emissions. Similarly, electric vehicle motion is produced from the transformation of electricity.  This electricity is generated from refined oil (in this example), and the refined oil is produced from crude oil.  Again, the latter two processes produce carbon dioxide ($CO_2$) and nitrogen oxide ($
NO_X$) emissions.  Next, the data in the process flow diagram shown in Fig. \ref{fig: Process Flow Diagram} is extracted to define the inputs $X=[x_1, \ldots, x_5]^T$, the outputs $Y=[y_1, \ldots, y_5]^T$, and environmental aspects $E=[e_1, \ldots, e_8]^T$
, vectors.
\begin{itemize}
\item $x_1$: How much the process ``Produce Refined Oil from Crude Oil" is executed. 
\item $x_2$: How much the process ``Generate Electricity from Refined Oil" is executed. 
\item $x_3$: How much the process ``Transform Electricity into EV Motion" is executed.
\item $x_4$: How much the process ``Produce the Gasoline from Crude Oil" is executed.  
\item $x_5$: How much the process ``Combust Gasoline to Generate ICV Motion" is executed.  
\item $y_1$: Amount of refined oil produced from Process 1. 
\item $y_2$: Amount of electricity produced from Process 2. 
\item $y_3$: Amount of distance travelled from Process 3. 
\item $y_4$: Amount of distance travelled from Process 5. 
\item $y_5$: Amount of gasoline produced from Process 4. 
\item $e_1$: Amount of $CO_2$ emitted to atmosphere. 
\item $e_2$: Amount of $NO_X$ emitted to atmosphere. 
\item $e_3$: Amount of Crude Oil taken from the Earth.  
\end{itemize}
\noindent Next, the technology matrix $A$ and the environmental matrix $B$ are derived from the arc labels shown in Fig. \ref{fig: Process Flow Diagram}. 

\begin{align}
A =  & \left[\begin{tabular}{p{0.55in}p{0.55in}p{0.55in}p{0.55in}p{0.55in}} 
  1      & -61.9    & 0       & 0       & 0    \\
  0      &  1       & -3.816  & 0       & 0    \\
  0      &  0       &  1      & 0       & 0    \\
  0      &  0       &  0      & 1       & 0    \\
  0      &  0       &  0      & -53.3   & 1 \end{tabular}\right] \\
 B = & \left[\begin{tabular}{p{0.55in}p{0.55in}p{0.55in}p{0.55in}p{0.55in}}  
1.030e-3   & 2.515e-3   & 0       & 2.48e-1   & 2e-4 \\
8.4e-4     & 2.237e-1   & 0       & 2.17e-1   & 5.4e-4  \\
-3.45      & 0          & 0       & 0         & -2.22 \end{tabular}\right]
\end{align}

\noindent Consequently, when $Y=[0, 0, 500, 0, 0]^T$, corresponding to a 500km trip in an EV, Eq. \ref{Equation 3} gives a result of 
\begin{align}
E=\begin{bmatrix}
1.264 \times 10^{2} \text{kg CO}_2 & 5.260 \times 10^{2} \text{g NO}_x & 4.075 \times 10^{5}
 \text{cc Crude Oil}
\end{bmatrix}^T
\end{align}
Similarly, when $Y=[0, 0, 0, 500, 0]^T$, corresponding to a 500km trip in an ICV, Eq. \ref{Equation 3} gives a result of
\begin{align}
E=\begin{bmatrix}
1.293 \times 10^{2} \text{kg CO}_2 & 1.228 \times 10^{2} \text{g NO}_x & 5.916 \times 10^{4} \text{cc Crude Oil}
\end{bmatrix}^T
\end{align}

\subsection{Life Cycle Analysis by Hetero-functional Graph Theory}
\label{Sec 3 (B)}
A similar analysis can be conducted with hetero-functional graph theory when its terminology is reconciled with the terminology used in process-based life cycle analysis. Returning to Def. \ref{Defn:D1 System Operand}, the system operands are the inputs and outputs of all of the processes shown in Fig. \ref{fig: Process Flow Diagram}.

\begin{align}
L &= \left\{
\begin{aligned}
&l_1 : \text{Refined Oil} \\
&l_2 : \text{Electricity} \\
&l_3 : \text{Distance travelled} \\
&l_4 : \text{Gasoline} \\
&l_5 : \text{CO}_2 \text{ Emissions} \\
&l_6 : \text{NO}_X \text{ Emissions} \\
&l_7 : \text{Crude Oil} \\
\end{aligned}
\right\}
\end{align}
Next, the system processes (in Defn. \ref{Defn:D2 System Process}) are all of the processes shown in Fig. \ref{fig: Process Flow Diagram}.
\begin{align}
P &= \left\{
\begin{aligned}
&
P_1 : \text{Produce Refined Oil from Crude Oil} \\
&P_2 : \text{Generate Electricity from Refined Oil} \\
&P_3 : \text{Transform Electricity into EV Motion} \\
&P_4 : \text{Produce Gasoline from Crude Oil} \\
&P_5 : \text{Combust Gasoline to Generate ICV Motion} \\
\end{aligned}
\right\}
\end{align}
Next, it is important to recognize that the system resources (in Defn. \ref{Defn:D3 System Resource}) are not mentioned in Fig. \ref{fig: Process Flow Diagram}.  Instead, they must be identified from the blocks in the block definition diagram in Fig. \ref{fig: BDD to Oil Vehicle Motion} or the swim lanes in Fig. \ref{fig: Activity Diagram Propel Vehicle from Crude Oil}.  In addition to these resources, a ``source" resource called Earth and a ``sink" resource called atmosphere are needed to define the location from which all natural resources are sourced and the location to which all emissions are sent respectively.  Consequently, 
\begin{align}
R &= \left\{
\begin{aligned}
&R_1 : \mbox{Oil Refinery} \\
&R_2 : \mbox{Oil Fired Power Plant} \\
&R_3 : \mbox{Electric Vehicle (EV)} \\
&R_4 : \mbox{Internal Combustion Vehicle (ICV)} \\
&R_5 : \mbox{Atmosphere} \\
&R_6 : \mbox{Earth} \\
\end{aligned}
\right\}
\end{align}
For simplicity of discussion, all of these resources are assumed to be stationary. (This assumption is revisited later in Sec.\ref{Sec:Discussion}). Therefore, the application of Defn. \ref{Defn:D4 Buffer} gives  $B = R$.  Therefore, the system's capabilities (as in Defn. \ref{Defn:D5 Capability}) become:
\begin{align}
{\cal E}_{S} &= \left\{
\begin{aligned}
&{\cal E}_{S_1} : \mbox{Oil Refinery produces Refined Oil from Crude Oil} \\
&{\cal E}_{S_2} : \mbox{Oil Fired Power Plant produces Electricity from Refined Oil} \\
&{\cal E}_{S_3} : \mbox{Electric Vehicle (EV) consumes Electricity for Motion} \\
&{\cal E}_{S_4} : \mbox{Oil Refinery produces Gasoline from Crude Oil} \\
&{\cal E}_{S_5} : \mbox{Internal Combustion Vehicle (ICV) consumes Gasoline for Motion}
\end{aligned}
\right\}
\end{align}

\noindent Next, the hetero-functional incidence tensor $\widetilde{\cal M}_\rho$ is determined from Defn. \ref{Defn:D6 HFIT -ve} and \ref{Defn:D7 HFIT +ve} and then matricized.  The resulting matrix has a size of $|L||B_S|x|{\cal E}_S|$ and its values must be inferred from the process flow diagram in Fig. \ref{fig: Process Flow Diagram}.  Note that because not all combinations of operand and buffer are necessary, the hetero-functional incidence matrix has rows filled entirely with zeros.  When these zero-filled rows are eliminated, the resulting hetero-functional incidence matrix $M$ becomes:
\begin{align}
\begin{tabular}{p{0.55in}p{0.55in}p{0.55in}p{0.55in}p{0.55in}} 
${\cal E}_{S_1}$ & ${\cal E}_{S_2}$ & ${\cal E}_{S_3}$ & ${\cal E}_{S_4}$ & ${\cal E}_{S_5}$ 
\end{tabular} \\
M = \begin{bmatrix}\begin{tabular}{p{0.55in}p{0.55in}p{0.55in}p{0.55in}p{0.55in}}
  1      & -61.9    & 0       & 0       & 0       \\
  0      &  1       & -3.816  & 0       & 0       \\
  0      &  0       & 1       & 0       & 0       \\
  0      &  0       & 0       & 1       & 0       \\
  0      &  0       & 0       & -53.3   & 1       \\
1.030e-3   & 2.515e-3   & 0       & 2.48e-1   & 2e-4 \\
8.4e-4  & 2.237e-1 & 0       & 2.17e-1 & 5.4e-4   \\
-3.45   & 0        & 0       & 0       & -2.22   \\
\end{tabular}\end{bmatrix} & 
\begin{tabular}{l}
$
l_1,b_1$: Refined Oil at Oil Refinery \\
$l_2,b_2$: Electricity in Oil-Fired Power Plant \\
$l_3,b_3$: Distance travelled at Electric Vehicle\\
$l_3,b_4$: Distance travelled at ICV\\
$l_4,b_1$: Gasoline at Oil Refinery\\
$l_5,b_5$: CO$_2$ Emissions at Atmosphere\\
$l_6,b_5$: NO$_x$ Emissions at Atmosphere \\
$l_7,b_6$: Crude Oil at Earth \\
\end{tabular}
\end{align}
Consequently, the Engineering System Net for this system is constructed from Defn. \ref{Defn:8 ESN} and depicted in Fig. \ref{Fig:PetriNet}.  Its state transition function follows from Defn.\ref{Defn:9 ESN-STF}.  The resulting simplified state transition function is stated in Eq. \ref{Eq:SimpleSTF}.  

\begin{figure}[H]
\centering
\includegraphics[width=1.0\linewidth]{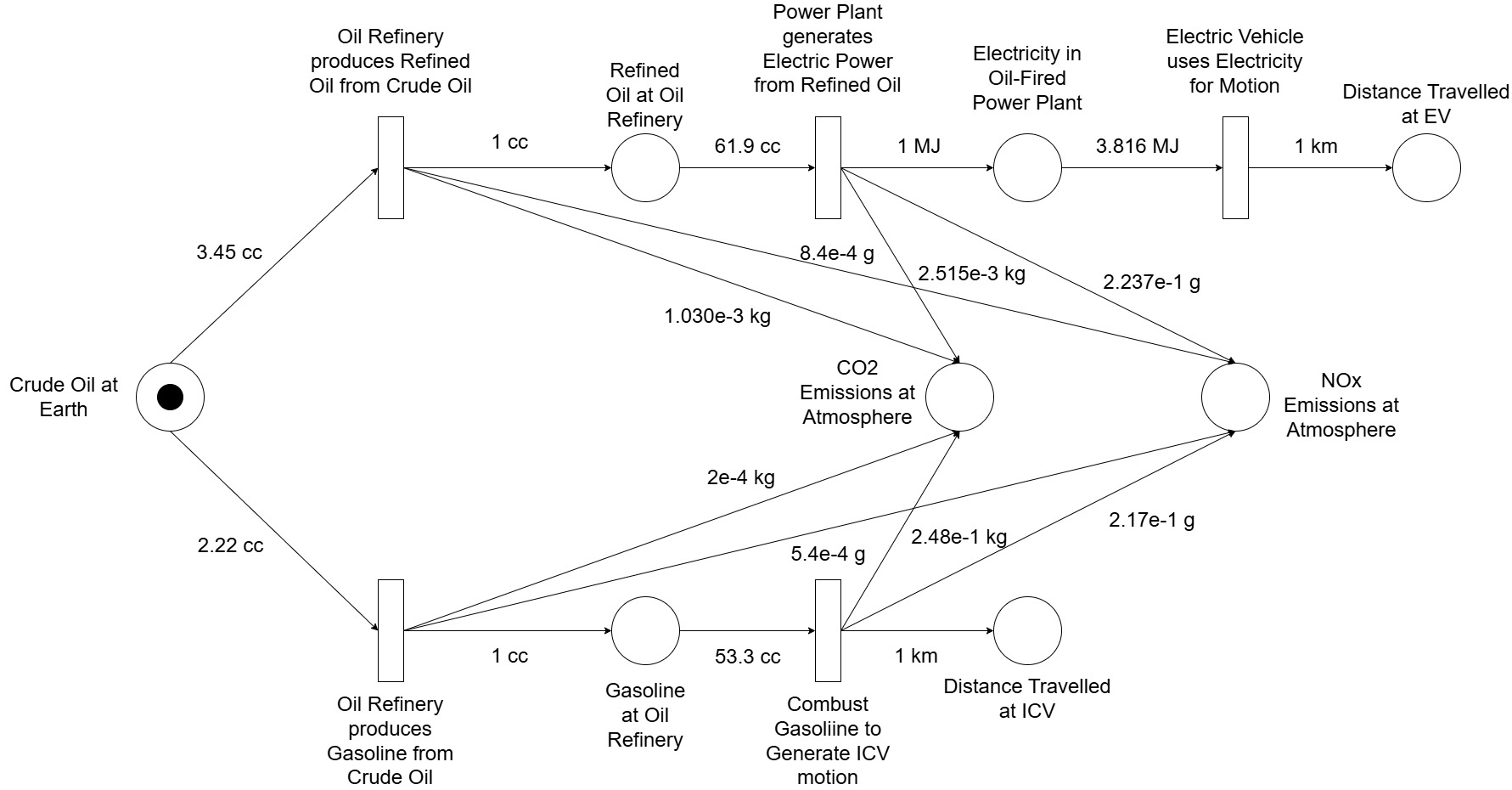}
\caption{Petri Net representing the Life Cycle Analysis (LCA) of the process flow from crude oil from Earth to the distance travelled by Electric Vehicles (EV) and Internal Combustion Vehicles (ICV), along with their respective emissions in the Atmosphere. The diagram illustrates the transformation stages of crude oil into refined oil, electricity, and gasoline, which then powers EV and ICV motion, resulting in specific emissions (CO$_2$, NO$_x$) at each stage.}
\label{Fig:PetriNet}

\end{figure}

To complete the life cycle analysis with hetero-functional graph theory, it is important to recognize that the process-based life cycle analysis only considers a simulation from an initial to a final state (with no time steps in between).  Therefore, the simulation horizon K=2.  Additionally, the simulation time step $\Delta T=1$.  Next, the process-based life cycle analysis only measures the \emph{change} in inputs used, outputs produced, and environmental aspects produced.  Consequently, the change in the marking vector $\Delta Q_{B}$ is also introduced.  
\begin{align}
\Delta Q_{B} = Q_{B}[k=2]- Q_{B}[k=1]
\end{align}
Under these conditions, and as a result, the equivalence between the simplified Engineering System Net state transition function and the process-based life cycle analysis models becomes clear:
\begin{align}\label{Eq:HFGT-PBLCA}
\Delta Q_{B}&={M}U \\\label{Eq:PBLCA}
\begin{bmatrix}
Y \\ E
\end{bmatrix} &= 
\begin{bmatrix}
A \\ B
\end{bmatrix}X
\end{align}
This result can be restated as a formal proof on the relationship between hetero-functional graph theory and process-based life cycle analysis.  
\begin{thm}\label{Thm:PCA_HFG}
Given an arbitrary process flow diagram, the associated life cycle analysis in Eq. \ref{Eq:PBLCA} is a formal special case of the engineering system net state transition function in Eq. \ref{Eq:ESNSTF1} and \ref{Eq:ESNSTF2}. 
\end{thm}
\begin{proof}
The proof is made directly with the introduction of the following special conditions.  
\begin{enumerate}
\item \textbf{Assumption 1:} Assume that each of the processes $p \in P$ are instantiated and then subsequently allocated exactly to one resource $r \in R$ so that there is a direct one-to-one mapping of each process to the associated capability in $e \in {\cal E}_S$ in hetero-functional graph theory.
\item \textbf{Assumption 2:} Next assume a simulation horizon of K=2. Also assume that the simulation time step $\Delta T$ = 1.
This results in Eq. \ref{Eq:HFGT-PBLCA}.
\item \textbf{Assumption 3:} Assume that all of the capabilities (i.e. transitions) ${\cal E}_S$ of the Engineering System Net State Transition Function occur instantaneously. $U^+[k]=U^-[k] \forall k$.  This results in the simplified state transition function in Eq. \ref{Eq:SimpleSTF}.  
\end{enumerate}
Under these special conditions, Eq. \ref{Eq:HFGT-PBLCA} is equal to Eq. \ref{Eq:PBLCA} when 
\begin
{align}
\Delta Q_{B}=\begin{bmatrix}
Y \\ E
\end{bmatrix}, \quad M = \begin{bmatrix}
A \\ B
\end{bmatrix}, \quad U &= X
\end{align}
\end{proof}

\section{Discussion:  Extending Life Cycle Analysis with MBSE \& HFGT}\label{Sec:Discussion}
Theorem 1 provides an analytical foundation upon which to discuss how model-based systems engineering and hetero-functional graph theory may be used to enhance the spatio-temporal resolution of process-based cycle analysis.  The discussion is organized around relaxing the three assumptions made above.
 
\subsection{Admitting the Allocation of System Processes to Multiple Resource Types \& Instances}

First, and perhaps most revealing, is the potential relaxation of Assumption 1 above so that \( P \neq \mathcal{E}_S \).  MBSE and HFGT make a special point not to conflate the generic set of system processes $P$ (Defn. \ref{Defn:D2 System Process} and specific set of system \emph{capabilities} ${\cal E}_S$ (Defn. \ref{Defn:D5 Capability}).  Here,it is understood that a system process is stated in a ``solution-neutral" manner\cite{Crawley:2015:00} while a system capability is stated in a ``solution-specific" manner.  Said differently, MBSE and HFGT recognize that a given process $p \in P$ can be executed by potentially two different resources, which may or may not be of the same type. 
 For example, in manufacturing, the process ``drill hole" can be executed by one of three different resources: $r_1 =$ drill press, $r_2 = $ milling machine, or $r_3 = $ lathe. Even though each of these resources carries out the same (solution-neutral) ``drill hole" process, the associated life cycle analysis will require different process weight values to account for the specific technology (or solution) that is carrying out that process.  In contrast, the process-based life cycle analysis literature demonstrates some cases \cite{Mathews:2015:00} where the term ``process" refers a solution-neutral ``process" (Defn. \ref{Defn:D2 System Process} and other cases \cite{X} where the term ``process" refers to solution-specific ``capability" (Defn. \ref{Defn:D5 Capability}).  The latter case is specifically required when the life cycle analyst recognizes that the process weight values vary significantly because of differences in underlying technology and/or operating procedures.  And so while ``process-based life cycle analysis" can be conducted on the basis of either (solution-neutral) processes or (solution-specific) capabilities, it is not immediately clear which is being conducted because the term ``process-based" is effectively overloaded.  Meanwhile, MBSE and HFGT avoid this conflation in terminology.  

The distinction between life cycle analyses based on processes versus capabilities is important. Of course, as stated above, the values of the process weights may differ from one capability to another, even though their underlying (solution-neutral) process is specific.  However, the use of capabilities in life cycle analysis brings about further complexities.  Consider the case where the ``drill hole'' process can be carried out by two \emph{identical} drill presses. At first glance, such a situation may appear to give identical life cycle analysis results. In reality, however, the two resource instances must exist in two different locations. Therefore, the life cycle analysis must be expanded to include \emph{transportation} processes that account for the distinct locations of the identical drill presses.  This situation can be elaborated in a more complex example.  Consider, for example, two homes that each use an identical propane-fueled boiler for residential heating. The first is in an urban neighborhood near a propane supply depot. The latter is in a rural neighborhood and is far from the propane supply depot. The life cycle analysis for 1BTU of heating will be different because the propane truck required to supply the two homes would have to drive further to the rural home than to the urban home.  This geographical effect is further exacerbated when considering that the (solution-neutral) transportation processes can have different process weights as solution-specific transportation capabilities.  Consider, for example, the transport of propane can occur via pipeline, rail, truck, or even an electrified truck; all with completely different process weights.  Not only is the life cycle analysis affected by the \emph{location} of the capabilities, but it is also affected by the specific technology carrying out the transportation processes between these locations.  

Many process-based life cycle analyses include only energy and matter \emph{conversion} capabilities and neglect the energy and matter \emph{transportation} capabilities.  Such an assumption may or may not be sound because the usage of capabilities \emph{necessitates} a specific geographic region where the capabilities may be either densely or sparsely situated.  It also necessitates an understanding of the technologies being used to carry out the transportation processes between the energy conversion capabilities.  Again, MBSE and HFGT make the role of transportation capabilities explicit.  Eq. \ref{Eq:HFGT-PBLCA} is revised to become:  

\begin{align}\label{Eq:HFGT-PBLCA2}
\Delta Q_{B}&=\begin{bmatrix}
M_{conversion} & M_{transportation}
\end{bmatrix}\begin{bmatrix}
U_{conversion} \\
U_{transportation}
\end{bmatrix}
\end{align}
Consequently, it is straightforward to evaluate whether the assumption
\begin{align}\label{Eq:Assumption}
M_{conversion}U_
{conversion} >> M_{transportation}U_{transportation}
\end{align}
Is indeed valid, or if transportation capabilities must be included.  Furthermore, because MBSE and HFGT provide an \emph{automated} means of constructing the heterofunctional incidence matrices $M_{conversion}$ and $M_{transportation}$, these matrices can be rapidly calculated for distinct geographies, and the assumption in Eq. \ref{Eq:Assumption} can be rapidly evaluated.  In the end, it is worthwhile knowing whether the results of the life cycle analysis are dominated by the process weights of the energy conversion processes, the specific geographic situation of these processes, or the technologies used by the transportation capabilities.   

\subsection{Admitting a Long Simulation Horizon $K > 2$}
It is also important to recognize that time itself can have an important role to play in life cycle analysis. For example, the carbon intensity of 1kW of electricity from the grid will vary tremendously from one moment to the next as the availability of solar irradiance and wind speed varies. Similarly, electricity markets continually revise the mix of coal, oil, and natural gas-fired power plants over the course of the day. Finally, if electricity is stored in a battery then carbon-free solar and wind energy can be captured during peak periods and then deployed to meet demand at another time.  Again, process-based life cycle analysis does not normally include the effect of storage process and/or capabilities.  In any of these conditions, a \emph{dynamic} life cycle analysis is required to understand life cycle aspects over the course of a simulation horizon $K > 2$. Such an analysis can be readily completed with the entire engineering system net state transition function in Eq. \ref{Eq:ESNSTF1} and \ref{Eq:ESNSTF2} rather than process-based life cycle analysis in Eq. \ref{Eq:PBLCA}.

\subsection{Admitting Processes with Varying Durations}
The previous subsection revealed that the life cycle aspects associated with a given operand do not just vary in space, but they also vary in time. Consequently, \textit{when} a process finishes matters. Consider, for example, an industrial site composed of a water pump, a water tank, and an electrolyzer. The industrial site is considering two water pump designs: an expensive one that fills the tank in 3 hours and a cheaper one that fills the tank in 6 hours. Naturally, the electrolyzer wishes to operate at times when the grid’s electricity is dominated by solar and wind generation. However, if the water tank is not filled fast enough, then the electrolyzer might be forced to operate at times when the electricity is dominated by carbon-intensive fuels rather than renewable energy resources.  Again, such an analysis can be readily completed with the entire engineering system net state transition function in Eq. \ref{Eq:ESNSTF1} and \ref{Eq:ESNSTF2} rather than process-based life cycle analysis in Eq. \ref{Eq:PBLCA}.

\section{Conclusion \& Future Work}\label{Sec:Conclusion}
This paper reconciles process-based life cycle analysis with the broader literature and vocabulary of model-based systems engineering and hetero-functional graph theory. In so doing, it facilitates the integration of the rich life cycle analysis literature into the design of complex systems more generally. Consequently, environmental sustainability is no longer an isolated consideration but rather one that is an inherent part of a system’s architecture and design. This reconciliation reveals the specific conditions under which model-based systems engineering (MBSE) and hetero-functional graph theory (HFGT) are equivalent to a life cycle analysis. Furthermore, MBSE and HFGT make a special point to distinguish between solution-neutral processes and solution-specific capabilities; which is not always apparent in the process-based life cycle analysis literature.  Ultimately, his paper proves that model-based systems engineering and hetero-functional graph theory are a formal generalization of process-based life cycle analysis. Consequently, this paper discusses how a life cycle analysis methodology can be extended.  In particular, it shows how MBSE and HFGT address systems that a.) contain processes with a diversity of processing times, b.) exhibit dynamic behavior over a long simulation horizon, c.) store physical quantities, d.) exhibit an explicit spatial distribution, and e.) utilize a diversity of transportation technologies to overcome that spatial distribution.   In all, the paper demonstrates how model-based systems engineering and hetero-functional graph theory may be used to enhance the spatio-temporal resolution of process-based life cycle analyses in a manner that aligns sustainability goals with system design objectives.  In future work, MBSE and HFGT are used to conduct spatio-temporal life cycle analyses on real-life case studies.

\section{Acknowledgement}
The authors are grateful for the gracious financial support of the National Energy Analysis Centre at the Commonwealth Scientific and Industrial Research Organisation in Australia.

\bibliographystyle{IEEEtran}
\bibliography{LIINESLibrary,LIINESPublications,LG_BG_HFGT_References,NirajReferences}

\end{document}